\documentclass[runningheads,12pt]{llncs}
\usepackage{amsmath,amssymb,amstext}
\usepackage{graphicx}
\usepackage{subfig}
\usepackage{cite}
\usepackage{url}
\usepackage{tabularx}
\usepackage[letterpaper,hmargin=1in,vmargin=1.25in]{geometry}%required margins
\urldef{\email}\path|aysajan@isy.liu.se,jan-ake.larsson@liu.se|
\newcommand{\keywords}[1]{\par\addvspace\baselineskip
\noindent\keywordname\enspace\ignorespaces#1}

\allowdisplaybreaks

\usepackage{tikz}
\usetikzlibrary{calc,shapes.misc}
\usepackage{latexsym}

\renewcommand{\epsilon}{\varepsilon}
\newcommand{\E}{\text E}

\newcommand{\Xtilde}{\widetilde{X}}
\newcommand{\Xhat}{\widehat{X}}

\newcommand{\mZ}{\mathcal{Z}}
\newcommand{\mA}{\mathcal{A}}
\newcommand{\mS}{\mathcal{S}}

\newcommand{\mH}{\mathcal{H}}
\newcommand{\mM}{\mathcal{M}}
\newcommand{\mT}{\mathcal{T}}
\newcommand{\mK}{\mathcal{K}}
\newcommand{\mF}{\mathcal{F}}
\newcommand{\mX}{\mathcal{X}}
\newcommand{\tr}{\delta}

\begin{document}
\mainmatter
\title{Direct Proof of Security of Wegman-Carter\\ Authentication  
  with Partially Known Key}
\titlerunning{Direct Proof of Security of Wegman-Carter Auth.\ with
  Partially Known Key}
\author{ Aysajan Abidin and Jan-{\AA}ke Larsson }
\authorrunning{A. Abidin and J.-{\AA}. Larsson}
% (feature abused for this document to repeat the title also on left hand pages)
\institute{Department of Electrical Engineering, \\
  Link\"oping University, SE-581 83 Link\"oping, Sweden\\
  \email}

\maketitle

\begin{abstract}
  Information-theoretically secure (ITS) authentication is needed in
  Quantum Key Distribution (QKD).  In this paper, we study security of
  an ITS authentication scheme proposed by Wegman\&Carter, in the case
  of partially known authentication key. This scheme uses a new
  authentication key in each authentication attempt, to select a hash
  function from an Almost Strongly Universal$_2$ hash function family.
  The partial knowledge of the attacker is measured as the trace
  distance between the authentication key distribution and the uniform
  distribution; this is the usual measure in QKD.  We provide direct
  proofs of security of the scheme, when using partially known key,
  first in the information-theoretic setting and then in terms of
  witness indistinguishability as used in the Universal Composability
  (UC) framework. We find that if the authentication procedure has a
  failure probability $\epsilon$ and the authentication key has an
  $\epsilon'$ trace distance to the uniform, then under ITS, the
  adversary's success probability conditioned on an authentic
  message-tag pair is only bounded by $\epsilon+|\mT|\epsilon'$, where
  $|\mT|$ is the size of the set of tags. Furthermore, the trace
  distance between the authentication key distribution and the uniform
  increases to $|\mT|\epsilon'$ after having seen an authentic
  message-tag pair. Despite this, we are able to prove directly that
  the authenticated channel is indistinguishable from an (ideal)
  authentic channel (the desired functionality), except with
  probability less than $\epsilon+\epsilon'$. This proves that the
  scheme is ($\epsilon+\epsilon'$)-UC-secure, without using the
  composability theorem.
  
\end{abstract}

\section{Introduction}
\label{sec:introduction}

Information-theoretically secure (ITS) message authentication codes
\cite{WC1,WC2} provide two users, Alice and Bob, with means to
guarantee authenticity and integrity of messages exchanged over an
insecure public channel. To achieve ITS (sometimes called
unconditional security) the schemes used need shared secret between
Alice and Bob.  This procedure is secure against any adversary, even
with unlimited computing and storage capability, provided that the key
is perfectly secret. Such schemes normally have high demand for fresh
secret key material, but even so they are used in some cryptographic
schemes; especially in ITS key agreement schemes such as Quantum 
Key Distribution (QKD) \cite{BB84,Ekert91}.  QKD needs ITS
authentication in order to thwart man-in-the-middle attacks
\cite{BB84,AL09,Abidin,PAL12}.

This paper addresses security of an ITS Authentication scheme
originally proposed by Wegman and Carter \cite{WC2}, in the case of
partially known key.  The scheme is based on secretly selecting a
function from a certain family of functions, details will be given in
what follows. The function is then used to create a message
authentication code, a tag, from the message. The important property
of the family in question is that revealing the output, the tag, from
one single use of a function does not reveal too much information on
which function is used. This is to prohibit an attacker from
identifying the function used, to generate a tag for another (forged)
message. However, revealing two tags for two different messages 
may reveal enough to generate a tag for a third, so the function 
cannot be reused. Several messages can be authenticated securely by
secretly selecting a new function for each desired authentication; we
will refer to this mode of operation as WCA.  Another is to hide the
output, by encrypting the tag using one-time pad encryption, but in
this paper, we only consider the WCA scheme.

The WCA scheme is ITS provided that the authentication key is 
uniformly distributed (or perfect). In practice, however, cryptographic 
keys are imperfect if partial information has leaked about them.
One example of this is QKD-generated keys, 
where an eavesdropper can extract some information on the key, tightly
restricted by security parameters of the system.  In this paper, we 
study security of the WCA scheme in the scenario where the key is
partially known to the adversary. We measure the adversary's partial
knowledge of the key as the trace distance between the distribution of
the key and the uniform distribution, as is done in QKD.  We should
stress that our analysis is not just restricted to QKD. The same
analysis applies \emph{whenever} the authentication scheme under
study is used with a key that has a small but non-zero trace distance
to the uniform.

\subsection*{Related work, and contribution of this paper}

The security of the WCA scheme as used in QKD was studied in \cite{CL}
where the observation was made that, for the WCA scheme with partially
known authentication key, an active attack is not always needed to
weaken the system. The attacker can, in essence, wait for a
beneficial moment and only launch an active (guessing) attack at that
moment. The paper also proposes a countermeasure to this that is
simple to implement.

A more recent paper \cite{Portmann} extends the security of the
WCA scheme to the Universally Composable (UC) framework, proving that
the scheme is UC-secure if the authentication key is perfectly secret.
In the same paper, the Composability Theorem \cite{Canetti} is used to
further extend the result to the case with partially known key, but
due to the complexity of the UC framework and the composability theorem, the
existence of the guessing attack mentioned above, and ultimately the
differences between questions of Confidentiality and Integrity, there
has been some discussion as to the meaning and appropriate statement
of this result \cite{OH12,RR12,HY12}.  

In this paper, we aim to resolve
the issue by providing upper bounds for failure probability, both
for the problem discussed in \cite{CL} and for witness indistinguishability as
used in the UC framework. This is done for the case of partially known
key using a direct proof, without using the Composability Theorem.  We
first show that, if the authentication procedure has a failure
probability $\epsilon$; the authentication key has an $\epsilon'$
trace distance to the uniform; and the adversary has seen a valid
message-tag pair, then the adversary's success probability of breaking
the authentication is only bounded by $\epsilon+|\mT|\epsilon'$, where
$|\mT|$ is the size of the tag space. This is \emph{significantly}
larger than what one would expect from the bound emerging from the UC framework.  
Despite this, we are able to prove directly that the authenticated channel is
distinguishable from an authentic channel (the desired functionality)
with probability less than $\epsilon+\epsilon'$.

The structure of the paper is as follows. Some background on
Universal hashing and its use in constructing ITS authentication
will be given in Section~\ref{sec:background}.  In
Section~\ref{sec:trace-distance}, we present some properties of subset
probability from distributions at nonzero trace distance from the
uniform, that are needed in the security proofs. 
The ITS security bound of the scheme when using partially known key is proved 
in Section~\ref{sec:ITS-security}, and the implications of the high bound
is discussed at the end of the section.  In Section~\ref{sec:UC-security}, 
we prove indistinguishability of the scheme from the ideal functionality 
when using partially known key. Section~\ref{sec:conclusions} concludes the paper. 
%The Appendix gives some background on Universal hashing and its use 
%in constructing ITS authentication, and also some properties of subset
%probability from distributions at nonzero trace distance from the
%uniform, that are needed in the security proofs. 

%%%%%%%%%%%%%%%%%%%%%%%%%%%%%%%%%%%%%%%%%%%
%%%%%%%%%%%%%%  New Section  %%%%%%%%%%%%%%%%%%%%%
%%%%%%%%%%%%%%%%%%%%%%%%%%%%%%%%%%%%%%%%%%%

\section{Background}
\label{sec:background}
In this section we present some necessary background that facilitates
understanding of the whole paper. First of all, we need to specify the
measure of partial knowledge to be used. \smallskip
\begin{definition}[The trace distance] This is also known as the
  variational distance or the statistical distance between two
  probability distributions $P_X$ and $P_X'$, and is
\begin{equation}\label{eq:trace-distance}
  \tr(P_X,P_X') = \tfrac12\sum_{x\in\mX}|P_X(x)-P_X'(x)|.
\end{equation}
\end{definition}
When we discuss security of a key in this paper, the following notion
will be used. \smallskip
\begin{definition}[Perfectness]\label{def:key-security}
  A key $k$ is called \emph{perfect} if it is uniformly distributed
  from the adversary's point of view; a key $k$ is called
  \emph{$\epsilon$-perfect}, if its distribution has an $\epsilon$
  trace distance to the uniform.
\end{definition}\smallskip

The family of functions used to create the tags are defined as
follows.  Let $\mM$ be the set of messages and $\mT$ be the set of
tags, both finite and $\mT$ typically much smaller than
$\mM$. Also, let $\mH$ be a set of functions from $\mM$ to $\mT$. The
appropriate set of functions to use in ITS authentication is the
following.\smallskip
\begin{definition}[Strongly Universal$_2$]
  The set $\mH$ is a \emph{Strongly Universal$_2$ (SU$_2$) hash
    function family} if \textbf{(a)} for any $m_1\in\mM$ and any
  $t_1\in\mT$ there exist exactly $|\mH|/|\mT|$ hash functions
  $h\in\mH$ such that $h(m_1)=t_1$, and \textbf{(b)} for any
  $m_2\in\mM$ (distinct from $m_1$) and any $t_2\in\mT$ (possibly
  equal to $t_1$), the fraction of those functions such that
  $h(m_2)=t_2$ is $1/|\mT|$.  If the fraction in (b) instead is at
  most $\epsilon$, the family $\mH$ is \emph{$\epsilon$-Almost
    Strongly Universal$_2$ ($\epsilon$-ASU$_2$)}.
\end{definition}\smallskip

When proving security of an authentication scheme, there are two
probabilities to bound: the probability of success in an
\emph{impersonation} attack, and the probability of success in a
\emph{substitution} attack. In an impersonation attack, the adversary
pretends to be a legitimate user and tries to generate the correct tag
for a (forged) message with no additional information, as would be
given by a valid message-tag pair. In a substitution attack, the
adversary intercepts a valid message-tag pair and tries to replace it
with a new message-tag pair. This latter attack is more powerful than
the former \cite{Johansson1}.

It is fairly straightforward to see that $\epsilon$-ASU$_2$ hash
functions can be used to construct unconditionally secure
authentication schemes in a natural way.  Let Alice and Bob share a
secret key $k$ to identify a hash function $h_k$ in a family $\mH$ of
$\epsilon$-ASU$_2$ hash functions from $\mM$ to $\mT$.  Alice sends her
message $m$ along with $t=h_k(m)$ to Bob. Upon receiving $m$ and $t$,
Bob verifies the authenticity of $m$ by comparing $h_k(m)$ with
$t$. If $h_k(m)$ and $t$ are identical, then Bob accepts $m$ as
authentic; otherwise, $m$ will be rejected.

Now, if Eve tries to impersonate Alice and sends $m^\prime$
without knowing the key $k$, or $h_k$, the best she can do is to
guess the correct tag for $m^\prime$. The probability of success in
this case is $1/|\mT|$.  Even if Eve waits until seeing a valid message-tag
pair $(m,t)$ from Alice, the probability of guessing the correct tag $t'$ 
for $m^\prime$ is at most $\epsilon$; cf. Def. 3(b). 
In other words, even seeing a valid message-tag pair does not increase
Eve's success probability above $\epsilon$.  Therefore, by 
using a family of $\epsilon$-ASU$_2$ hash functions with suitably
chosen $\epsilon$, one can achieve unconditionally secure
message authentication. 

In this scheme, however, a key cannot be used more than once,
because a repeated use of the same key may give Eve enough information
to forge a valid message-tag pair; Def.\ 3 does not say anything about
set sizes for three message-tag pairs. Therefore, in the mode of
operation considered here, WCA, a new secret key is used for each
authentication. The key length for
typical known families of $\epsilon$-ASU$_2$ hash functions is
logarithmic in the message length $\log|\mM|$ \cite{Stinson1,Stinson2,
Stinson3,Stinson4,AticiStinson,Hugo,Hugo1,Johansson,Johansson1,
Boer,AL11}, where $\log$ denotes the binary
logarithm. Hence, the key-consumption rate of WCA is logarithmic in the
message length.

%%%%%%%%%%%%%%%%%%%%%%%%%%%%%%%%%%%%%%%%%%%
%%%%%%%%%%%%%%  New Section  %%%%%%%%%%%%%%
%%%%%%%%%%%%%%%%%%%%%%%%%%%%%%%%%%%%%%%%%%%

\section{Probabilities of sets with 
  non-uniform underlying distribution}
\label{sec:trace-distance}

In what follows, we will need some simple results of probabilities of
subsets of key values, or hash functions, when the key is
$\epsilon$-perfect.  In general we denote the probability of a subset
of values $\mX'\subseteq\mX$ by
\begin{align*}
  P_X(\mX') = \sum_{x\in\mX'}P_X(x).
\end{align*}
First we note a simple property of the probability of a subset of
$\mX$, when the distribution has a nonzero trace distance to the
uniform distribution.

\begin{lemma}\label{lemma:subset}
  If the trace distance between $P_X$ and the uniform distribution is
  $\epsilon$, then for any subset $\mX'\subseteq\mX$,
  \begin{equation}\label{eq:subset}
    \Big|P_X(\mX')-\frac{|\mX'|}{|\mX|}\Big| \le \epsilon.
  \end{equation}
  Also, there are subsets that reach the bound.
\end{lemma}
\begin{proof}
  With $\mX_+:=\{x\in\mX:P_X(x)>1/|\mX|\}$ and
  $\mX_-:=\{x\in\mX:P_X(x)<1/|\mX|\}$, it is straightforward to see
  that
  \begin{equation}
    \begin{split}
      \epsilon=\tfrac12\sum_{x\in\mX}\Big|P_X(x)-\frac1{|\mX|}\Big|
      =P_X(\mX_+)-\frac{|\mX_+|}{|\mX|}=\frac{|\mX_-|}{|\mX|}-P_X(\mX_-).
    \end{split}
  \end{equation}
  Now, for any subset $\mX'\subseteq\mX$, we have
  \begin{equation}
    \begin{split}
      P_X(\mX')-\frac{|\mX'|}{|\mX|}
      \le P_X(\mX'\cap\mX_+)-\frac{|\mX'\cap\mX_+|}{|\mX|}
      \le P_X(\mX_+)-\frac{|\mX_+|}{|\mX|}=\epsilon
    \end{split}
  \end{equation}
  and also 
  \begin{equation}
    \begin{aligned}
      \frac{|\mX'|}{|\mX|}-P_X(\mX') 
      \le \frac{|\mX'\cap\mX_-|}{|\mX|}-P_X(\mX'\cap\mX_-)
      \le \frac{|\mX_-|}{|\mX|} - P_X(\mX_-)=\epsilon. 
    \end{aligned}
  \end{equation}
  This proves the inequality, and the subsets $\mX'=\mX_+$ and
  $\mX'=\mX_-$ both reach the bound.\hfill$\square$
\end{proof}

From this lemma follows a bound for the conditional probability of an
even smaller subset of $\mX$, when the distribution has a nonzero
trace distance to the uniform distribution.  We will use this later
when discussing security with preexisting partial knowledge and
additional gained knowledge in the message exchange.
\begin{theorem}\label{lemma:sub-subset}
  If the trace distance between $P_X$ and the uniform distribution is
  $\epsilon$, then for any subsets $\mX''\subseteq\mX'\subseteq\mX$,
\begin{equation}\label{eq:sub-subset}
  \Big|P_X(\mX''\,|\,\mX') - \frac{|\mX''|}{|\mX'|}\Big| 
  \le \frac{|\mX|}{|\mX'|}\epsilon.
\end{equation}
Also, there are subsets which reach the bound.
\end{theorem}
\begin{proof}
  The conditional probability can be written
\begin{equation}\label{eq:sss2}
\begin{aligned}
  P_X(\mX''\,|\,\mX') &= \frac{P_X(\mX'')}{P_X(\mX')} 
  = \frac{P_X(\mX'')}{P_X(\mX'') + P_X(\mX'\setminus\mX'')}
  = \left(1+\frac{P_X(\mX'\setminus\mX'')}{P_X(\mX'')}\right)^{-1}.
\end{aligned}
\end{equation}
To bound this from above, we need an upper bound for $P_X(\mX'')$ 
and a lower bound for $P_X(\mX'\setminus\mX'')$, both of which can 
be obtained using Lemma \ref{lemma:subset},
\begin{equation}\label{eq:sss3}
  % \frac{|\mX''|}{|\mX|} - \epsilon \le 
  P_X(\mX'')\le \frac{|\mX''|}{|\mX|} + \epsilon;\quad
  P_X(\mX'\setminus\mX'') \ge \frac{|\mX'\setminus\mX''|}{|\mX|}
  - \epsilon.%, 
\end{equation}
These give us the upper bound
\begin{equation}\label{eq:sss5}
\begin{aligned}
  P_X(\mX''\,&|\,\mX')
  = \left(1+\frac{P_X(\mX'\setminus\mX'')}{P_X(\mX'')}\right)^{-1} 
  \le\left(1+\frac{\frac{|\mX'\setminus\mX''|}{|\mX|} -
      \epsilon}{\frac{|\mX''|}{|\mX|} + \epsilon} \right)^{-1} =
  \frac{|\mX''|}{|\mX'|} + \frac{|\mX|}{|\mX'|}\epsilon.
\end{aligned}
\end{equation}
Similarly, from Lemma \ref{lemma:subset} we also know that 
\begin{equation}\label{eq:sss6} 
  P_X(\mX'')\ge\frac{|\mX''|}{|\mX|} - \epsilon;\quad
  P_X(\mX'\setminus\mX'') \le \frac{|\mX'\setminus\mX''|}{|\mX|} 
  + \epsilon.
\end{equation}
These give us the lower bound
\begin{equation}\label{eq:sss8}
\begin{aligned}
  P_X(\mX''\,&|\,\mX') = \left(1+\frac{P_X(\mX'\setminus\mX'')}
    {P_X(\mX'')}\right)^{-1} 
  \ge\left(1+\frac{\frac{|\mX'\setminus\mX''|}{|\mX|} + \epsilon}
    {\frac{|\mX''|}{|\mX|} - \epsilon} \right)^{-1} =
  \frac{|\mX''|}{|\mX'|} - \frac{|\mX|}{|\mX'|}\epsilon.
\end{aligned}
\end{equation}
This proves the inequality. The bound can be reached in several ways,
for example when $(\mX_+\cup\mX_-)\subseteq\mX'$ and
$\mX''=\mX_+$.\hfill$\square$
\end{proof}

Using the above theorem, we can derive a bound for the trace 
distance of the conditional distribution of $x$ on a subset
$\mX'\subseteq\mX$. This will be useful when discussing 
trace distance in relation to security later.
\begin{theorem}\label{lemma:cor}
  If the trace distance between $P_X$ and the uniform distribution is
  $\epsilon$, then given a subset $\mX'\subseteq\mX$, the conditional
  distribution of $x$ on $\mX'$ has trace distance to the uniform (on
  $\mX'$) that is bounded by
  \begin{equation}\label{eq:cor}
    \tfrac12\sum_{x\in\mX'}\Big|P_X(x\,|\,\mX')-\frac1{|\mX'|}\Big| 
    \le \frac{|\mX|}{|\mX'|}\epsilon.
  \end{equation}
  For certain subsets $\mX'$, the bound is reached.
\end{theorem}
\begin{proof}
It is straightforward to see that 
\begin{equation}\label{eq:cor1}
\begin{aligned}
\tfrac12\sum_{x\in\mX'}&\Big|P_X(x\, |\, \mX') - \frac1{|\mX'|}\Big| 
= P_X(\mX_+\cap \mX' \,|\, \mX') - \frac{|\mX_+\cap\mX'|}{|\mX'|}
\le \frac{|\mX|}{|\mX'|}\epsilon,
\end{aligned}
\end{equation}
where the inequality follows from Theorem~\ref{lemma:sub-subset}. 
The bound is reached when $\mX_+\cup\mX_-\subseteq\mX'$.\hfill$\square$
\end{proof}

%%%%%%%%%%%%%%%%%%%%%%%%%%%%%%%%%%%%%%%%%%%
%%%%%%%%%%%%%%  New Section  %%%%%%%%%%%%%%
%%%%%%%%%%%%%%%%%%%%%%%%%%%%%%%%%%%%%%%%%%%

\section{Information-theoretic security with partially known key}
\label{sec:ITS-security}

In this section we analyse security of the authentication 
scheme under study in information-theoretic setting, in the 
scenario where the key has a small but non-zero trace distance 
to the uniform. The WCA scheme uses $\epsilon$-ASU$_2$ hashing, 
and is $\epsilon$-secure, meaning that the probability of success 
in a substitution attack is bounded above by $\epsilon$, if the
authentication key is uniformly distributed (perfect).  We will now
analyse what happens when this is not the case, when the trace
distance to the uniform is nonzero. This means that the authentication
key is a random variable $K$ to Eve, and we use $\epsilon'$ to denote
its trace distance to the uniform.

We will start by giving an example of how large Eve's probability for
a successful substitution attack can become, even when using a SU$_2$
family. Since we are talking about a substitution attack, we need to
calculate the probability conditioned on Eve having seen a message-tag
pair $(m,t)$ from Alice.  One possible distribution is
\begin{equation}\label{eq:key-distr}
  P_K(k)=
  \begin{cases}
    \frac1{|\mK|}+\epsilon',&\text{if }k\in\mK_+=\{k_+\}\\
    \frac1{|\mK|}-\epsilon'\frac1{|\mK_-|},
    &\text{if } k\in\mK_-\\
    \frac1{|\mK|},&\text{otherwise}.
  \end{cases}
\end{equation}
This has trace distance $\epsilon'$ to the uniform. If
$\epsilon'>1/|\mK|$, the set $\mK_-$ must contain more than one
value. (Compare with the distribution used in \cite{CL} where
$P_K(k)=0$ if $k\in\mK_-$; $P_K(k)=1/(|\mK|-|\mK_-|)$ if 
$k\in\mK_+=\mK\setminus\mK_-$; and $\epsilon'=|\mK_-|/|\mK|$.) 
It is easy to see that Eve's probability for success, without more 
information on $K$, is maximal if she chooses $t_\E=f_{k_+}(m_\E)$ 
and $m_\E$ is such that $t_\E\neq{}f_{k_-}(m_\E)$ for all $k_-\in\mK_-$. 
Since the hash function family is SU$_2$, $|\{k:f_k(m_\E)=t_\E\}|=|\mK|/|\mT|$, 
and this set contains $k_+$ but excludes $\mK_-$ so that
\begin{equation}
  \begin{split}
    \Pr\big\{f_{K}(m_\E)=t_\E\big\}
    =\frac1{|\mK|}+\epsilon'+\Big(\frac{|\mK|}{|\mT|}-1\Big)\frac1{|\mK|}
    =\frac{|\mK|}{|\mT|}\frac1{|\mK|}+\epsilon'
    =\frac1{|\mT|}+\epsilon'.
  \end{split}\label{eq:2}
\end{equation}
It is also easy to see that Eve's probability for success increases if
she sees a valid message-tag pair $(m,t=f_K(m))$. Eve's gain will now
depend on $m$, and her gain is maximal if both $f_{k_+}(m)=t$ and
$f_{k_-}(m)=t$ for all $k_-\in\mK_-$, so that
\begin{equation}
  \Pr\big\{f_{K}(m) =t\big\}
  =\frac{|\mK|}{|\mT|}\frac1{|\mK|}+\epsilon'
  -|\mK_-|\epsilon'\frac1{|\mK_-|}
  =\frac1{|\mT|}.
\end{equation}
If $\epsilon'$ is small, there will exist such messages $m$. Since
the hash function family is SU$_2$,
$|\{k:f_k(m_\E)=t_\E\wedge{}f_k(m)=t\}|=|\mK|/|\mT|^2$, and again this
set contains $k_+$ but excludes $\mK_-$. Therefore
\begin{equation}
  \begin{split}
    \Pr\big\{f_{K}(m_\E)=t_\E\,\big|\, f_{K}(m)=t\big\}
      &=\frac{\Pr\big\{f_{K}(m_\E)=t_\E\wedge f_{K}(m)=t\big\}}
      {\Pr\big\{f_{K}(m)=t\big\}}\\
      &=\frac{\frac{|\mK|}{|\mT|^2}\frac1{|\mK|}+\epsilon'}{\frac1{|\mT|}}
      =\frac{\frac1{|\mT|^2}+\epsilon'}{\frac1{|\mT|}}
      =\frac1{|\mT|}+|\mT|\epsilon'.
  \end{split}
  \label{eq:1}
\end{equation}
Note that this is an equation, not an inequality. Before seeing
$(m,t)$ Eve's probability of a successful message insertion attack
equals $1/|\mT|+\epsilon'$. After seeing $(m,t)$, Eve's probability
of a successful substitution attack \emph{equals}
$1/|\mT|+|\mT|\epsilon'$.

This might be taken as cause for alarm, but one should note that this
is message-dependent: not all message-tag pairs $(m,t)$ will cause such an
increase. 
It was pointed out already in \cite{CL} that the message and used key
value may be such that Eve may have this unexpectedly high probability
of success. On the other hand, in some situations (here, when
$f_{k_+}(m)\neq t$), Eve will instead find out that her most likely
key value was, in fact, not used, and that she must remove it from the
set of possible key values. In this case, the information she had
becomes unusable; she will have lost information. But, importantly,
Eve can find out if there was a gain or not, before performing an
active (guessing) attack, by using her distribution of $K$ and the
received message-tag pair from Alice. Eve then only performs an active
attack if her success probability has increased (sufficiently, see
\cite{CL}). From Alice's point of view, the probability of having her
message-tag pair \emph{and} a successful attack from Eve is
$1/|\mT|+\epsilon'$, but this probability is \emph{per round}, not
per guess (by Eve). Eve does not need to reveal herself by guessing
frequently; she can wait for the beneficial case where her success
probability is high \cite{CL}. 

Therefore, there is a clear need for an upper bound for the success
probability in this situation. For general $\epsilon$-ASU$_2$-based
authentication, the following theorem holds.

\begin{theorem} \label{thm-WCA} \emph{(Bound for guessing probability
    with partially known key.)} Consider the WCA scheme based on
  $\epsilon$-ASU$_2$ hashing.  If the authentication key is
  $\epsilon'$-perfect (as random variable $K$ to the adversary), the
  probability of a successful message insertion is bounded by
  \begin{equation}
    \Pr\big\{f_K(m_\E)=t_\E \big\}\le\frac1{|\mT|}+\epsilon'.
  \end{equation}
  If in addition the adversary has access to a valid message-tag pair
  $(m,t)$, the probability of a successful substitution is bounded by
  \begin{equation}\label{eq:thm-WCA}
    \Pr\big\{f_K(m_\E)=t_\E \,\big|\, f_K(m)=t\big\} 
    \le \epsilon+|\mT|\epsilon'.
  \end{equation}
\end{theorem}

\begin{proof}
  The first inequality is obtained by applying
  Lemma~\ref{lemma:subset} to the set
  $\{k\in\mK:f_k(m_\E)=t_\E\}$. Since the hash function family is
  $\epsilon$-ASU$_2$ (Def. 3(a)), this set has the size $|\mK|/|\mT|$,
  and
  \begin{equation}\label{WCA-1}
    \Pr\big\{f_K(m_\E)=t_\E\big\}
    \le\frac{|\mK|}{|\mT|}\frac1{|\mK|}+\epsilon'
    =\frac1{|\mT|}+\epsilon'.
  \end{equation}
  To bound the probability that the adversary sees $(m,t)$ \emph{and}
  performs a successful substitution attack, we denote the subset of
  authentication key values that gives $(m,t)$ by
  \begin{equation}\label{eq:WCA-2}
  \mK' = \{k\in\mK: f_k(m)=t\}, 
  \end{equation}
  and where the attack succeeds by
  \begin{equation}
    \mK''=\{k\in\mK: f_k(m_\E) = t_\E \wedge f_k(m) = t \}.
    \label{eq:WCA-3}
  \end{equation}
  We know from Def. 3 that $|\mK'|=|\mK|/|\mT|$ and that
  $|\mK''|=\epsilon|\mK|/|\mT|$.  So using Theorem
  \ref{lemma:sub-subset}, we have
  \begin{equation}\label{eq:WCA-4}
    \begin{split}
      \Pr\big\{f_K(m_\E)&=t_\E \,\big|\, f_K(m)=t\big\} =P_K(\mK''\,|\,\mK')
      \le \frac{|\mK''|}{|\mK'|} + \frac{|\mK|}{|\mK'|}\epsilon' \le \epsilon+|\mT|\epsilon'.
    \end{split}
  \end{equation}
\hfill$\square$
\end{proof}

This theorem tells us that the previous example really is a worst-case
scenario, so that the upper bound for Eve's success probability after
seeing a message-tag pair is $\epsilon+|\mT|\epsilon'$. Conversely,
the example shows that the bound is sharp: there are situations where
the bound is reached, so the bound cannot be lowered if one wants
information-theoretic security.

In the Universal Composability framework (to be discussed below), the
relevant figure of merit is the trace distance to the uniform
distribution, and not the guessing probability as given above. And
also the trace distance increases by the same amount, in the
beneficial case for Eve. The key is still random to Eve, but the
distribution conditioned on her new knowledge that $h_K(m)=t$ has a
larger trace distance to the uniform.  A uniform distribution
conditioned on $h_K(m)=t$ would be constant at $|\mT|/|\mK|$ (the set
of still possible keys has the size $|\mK|/|\mT|$), but in our
example, if both $f_{k_+}(m)=t$ and $f_{k_-}(m)=t$ for all
$k_-\in\mK_-$,
\begin{equation}
  \label{eq:3}
  \begin{split}
    P_K(k_+\,|\,h_K(m)=t)&=\frac{\Pr\{K=k_+\wedge
      h_K(m)=t\}}{P\{h_K(m)=t\}}
    =\frac{P_K(k_+)}{\Pr\{h_K(m)=t\}}\\
   &=\frac{\frac1{|\mK|}+\epsilon'}{\frac1{|\mT|}}
    =\frac{|\mT|}{|\mK|}+|\mT|\epsilon'.
  \end{split}
\end{equation}
This forces the conditional distribution of the key to have a high
trace distance to the uniform. As before, the example gives the
worst-case scenario, and an upper bound for this trace distance is
given by the following theorem.
\begin{theorem}\label{thm-WCA-cor}
  \emph{(Bound for trace distance with partially known key.)}
  Consider the WCA scheme based on $\epsilon$-ASU$_2$ hashing.  
  If the authentication key is $\epsilon'$-perfect (as random variable 
  $K$  to the adversary), and the adversary has access to a valid
  message-tag pair $(m,t)$, then the trace distance from the
  conditional probability to the uniform is bounded by
   \begin{equation}\label{eq:WCA-cor}
     \tfrac12\sum_{k:f_k(m)=t}\bigg|P_K(k\,|\,f_K(m)=t)
     -\frac1{|\{k:f_k(m)=t\}|}\bigg| 
     \le  |\mT|\epsilon'.
  \end{equation}
\end{theorem}

\begin{proof}
  We use $\mK' = \{k\in\mK: f_k(m)=t\}$ and immediately obtain the
  bound from Theorem~\ref{lemma:cor}:
   \begin{equation}\label{eq:WCA-cor-1}
     \tfrac1{2}\sum_{k\in\mK'}\Big|P_K(k\,|\,\mK')
     -\frac1{|\mK'|}\Big| 
     \le \frac{|\mK|}{|\mK'|}\epsilon'
     =|\mK|\frac{|\mT|}{|\mK|}\epsilon'=|\mT|\epsilon'.
  \end{equation}
\hfill$\square$
\end{proof}

Again, the bound is sharp because of the example: there are situations
where the bound is reached, so the bound cannot be lowered if one
wants information-theoretic security. Note that, again, that this
depends on $(m,t)$, and a similar argument as that used above applies to
Eve's success rate. The upper bound is only reached in beneficial
situations (for Eve). 

The example shows that the bounds cannot be lowered, but are only
reached for certain $(m,t)$. This means that
the notion of ITS used here is ill suited for the situation. It works
well for perfect keys, because there, the probability of a successful
attack is equally bounded, with a low bound. It is clear that the situation is the
same whether one looks at guessing probability or trace distance;
there is a substantial, but non-constant increase. This is the reason
to turn to the notion of indistinguishability, which is better suited
for this situation.

\section{Indistinguishability from Ideal Authentication}
\label{sec:UC-security}
The notion of witness indistinguishability was first introduced in
\cite{FS}.  Here, we use the indistinguishability notion to prove
that, despite the substantially high bound for ITS, the WCA scheme
with an $\epsilon'$-perfect key is indistinguishable from the ideal
authentication, except with probability $\epsilon+\epsilon'$.  As a
natural consequence, Universally Composable (UC) security of the WCA
scheme with an $\epsilon'$-perfect key directly follows from our proof
of indistinguishability.

\begin{figure}[t]
\begin{center}
\beginpgfgraphicnamed{ideal}
\begin{tikzpicture}[scale=.9,thick,every node/.style={minimum height=.6cm}]
  \scriptsize
  \draw (-2.6,0) node[fill=black!10,minimum width=2cm,minimum height=.8cm]{};%
  \draw (2.6,0) node[fill=black!10,minimum width=2cm,minimum height=.8cm]{};%
  \draw (-3,0) node[draw,minimum width=1cm,fill=white](a) {Alice};%
  \draw (3,0)node[draw,minimum width=1cm,fill=white](b) {Bob};%
  \draw (0, 0)node[draw,minimum width=3cm,fill=white](f) {$\mF$};%

  \tiny
  \draw[->] (a) to node[above]{$m$} (f);%Alice gives m to F
  \draw ($(f.south)+(-1,0)$) to[bend right=45] node[below left]{$m$}
  ++(.5,-.5);%
  \draw[dashed] ($(f.south)+(-1,0)$)++(.5,-.5) to ++(1,0);%
  \draw[->] ($(f.south)+(-1,0)$)++(.5,-.5)++(1,0) to[bend right=45]
  node[below right]{$m'$} ($(f.south)+(1,0)$);%
  \draw[->] (f) to node[above]{$m,\bot$} (b);%      
\end{tikzpicture}
\qquad\quad
\begin{tikzpicture}[scale=.9,thick,every node/.style={minimum height=.6cm}]
  \scriptsize%
  \draw (-2,0) node[fill=black!10,minimum width=3cm,minimum
  height=.8cm]{};%
  \draw (2,0) node[fill=black!10,minimum width=3.1cm,minimum
  height=.8cm]{};%
  \draw (-3,0) node[draw,minimum width=1cm,fill=white](a) {Alice};%
  \draw (3,0)node[draw,minimum width=1cm,fill=white](b) {Bob};%
  \draw (0, 1)node[draw,rounded rectangle, minimum width=1.3cm](k)
  {Key};%
  \draw (-1,0)node[draw,rectangle,minimum width=1cm,fill=white](T)
  {TAG};%
  \draw (1,0)node[draw,rectangle,minimum width=1cm,fill=white](V)
  {VRFY};%
  \draw (0,.5) node[draw,dotted,minimum width=3cm,minimum
  height=1.7cm]{} +(1.2,.75)node{WCA}; %

  \tiny%
  \draw[->] (a) to node[above]{$m$} (T);%Alice gives m to TAG
  \draw (T.south) to[bend right=45] node[below left]{$(m,t)$}
  ++(.5,-.5);%
  \draw[dashed] (T.south)++(.5,-.5) to ++(1,0);%
  \draw[->] (T.south)++(.5,-.5)++(1,0) to[bend right=45] node[below
  right]{$(m',t')$} (V.south);%
  \draw[->] (V) to node[above]{$m',\bot$} (b);%

  % Alice and Bob get k from key source
  \draw[->] (k.south) to node[above right,at end]{$k$} (T.north);
  \draw[->] (k.south) to node[above left,at end]{$k$} (V.north);

  \end{tikzpicture}
  \endpgfgraphicnamed
  \caption{On the left is the ideal functionality: Alice gives her
    message $m$ to the ideal functionality $\mF$, which delivers it to
    Bob if it has not been modified on the channel ($m'=m$), otherwise
    the symbol $\bot$ is delivered. On the right is the real
    implementation in WCA: Alice uses the tag generation algorithm TAG
    to generate a tag $t$ and sends $(m,t)$. At the receiving end, Bob
    uses the verification algorithm VRFY to check if the received
    $(m',t')$ is a valid pair. If not, the symbol $\bot$ is
    delivered. }
  \label{fig:functionality}
\end{center}
\end{figure}
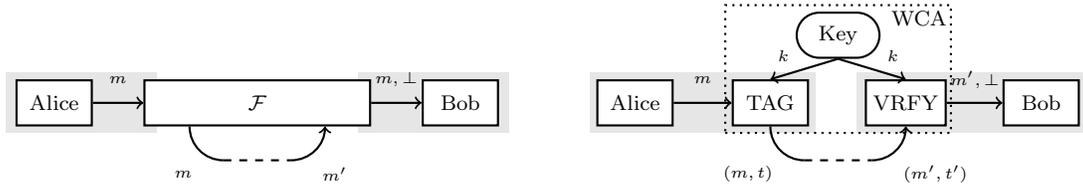

The ideal functionality of authentication, an \emph{authentic} channel
$\mF$, connects Alice and Bob in such a way that Bob can be certain
that any message output from the channel was sent by Alice. If the
message was modified on the channel, the symbol $\bot$ is delivered,
see Fig.~\ref{fig:functionality}. In other words, messages received from
$\mF$ are either authentic or blocked, and so \emph{cannot} be
successfully modified or substituted. Note that there is no
confidentiality requirement, so the message can be read by anyone.
The real implementation of authentication in the WCA scheme has three
components, as depicted in Fig.~\ref{fig:functionality}: a tag generation
algorithm TAG, a verification algorithm VRFY, and a key source. Both
TAG and VRFY use the same key.  From an input message $m$, Alice uses
TAG to compute a message-tag pair $(m,t)$ where $t=f_k(m)$ and $f_k$
is a hash function from an $\epsilon$-ASU$_2$ family identified by
$k$. Bob uses VRFY to verify a received message-tag pair $(m',t')$,
and VRFY outputs $m'$ if $f_k(m')=t'$ (for example if $m'=m$ and
$t'=t$), otherwise $\bot$.

The distinguisher (in UC terminology, the \emph{environment}) $\mZ$
should not be able to distinguish the two systems, except with low
probability. It can attempt to distinguish the two by controlling the
input to the system (the message $m$), and the output from the channel
$(m',t')$. The systems should be indistinguishable even under the
presence of an \emph{adversary} $\mA$, and it is sufficient to
consider the system under an adversary completely controlled by the
environment \cite{Canetti}, a \emph{dummy adversary} that only
forwards the desired channel output from the environment. As is, the
systems are trivially distinguishable because of the lack of a tag in
the ideal system. We therefore add a \emph{simulator} $\mS$ to the
ideal functionality, that adds a tag $t$ that is generated from $m$
using the appropriate key and hash function to make it
indistinguishable from the real case, and strips off any received tag
$t'$ after the channel. The name simulator also alludes to simulating
the adversary, and is especially simple when simulating the dummy
adversary.

We now want to ensure that the environment $\mZ$ cannot distinguish
between the two cases \emph{(a)} it is interacting with $\mA$ and
participants running the WCA scheme or \emph{(b)} it is interacting
with $\mS$ and participants running $\mF$, except with low probability
(see Fig.~\ref{fig:setup1}).  Perhaps we should point out that the
description here differs slightly from that of \cite{Portmann}. The
WCA scheme is resolved in somewhat finer detail and is separated from the
participants, and the ideal functionality is that of an authentic
channel rather than an immutable but blockable channel.  This is done
solely for the purpose of clear comparison of the real and the ideal
cases, and does not affect the results of the security evaluation.
Now, having set the stage, we can state our main theorem.

\begin{figure}[t]
\begin{center}
\begin{tikzpicture}[scale=.9,thick,every node/.style={minimum height=.6cm}]
  \scriptsize
   
% Environment, simulator, F_auth, Alice and Bob.
  \draw (0,-3) node[draw,minimum width=6.5cm] (z) {$\mZ$};%
  \draw (-2.6,0) node[fill=black!10,minimum width=2cm,minimum height=.8cm]{};%
  \draw (2.6,0) node[fill=black!10,minimum width=2cm,minimum height=.8cm]{};%
  \draw (-3,0) node[draw,minimum width=1cm,fill=white](a) {Alice};%
  \draw (3,0)node[draw,minimum width=1cm,fill=white](b) {Bob};%
  \draw (0,-1.5)node[draw,minimum width=3cm](s) {$\mS$};%
  \draw (0, 0)node[draw,minimum width=3cm,fill=white](f) {$\mF$};%
  \draw (-2.8, -1.5)node[draw,rounded rectangle, minimum width=1.3cm](k)
  {Key};%

  \tiny
  \draw[->] (z.west) to[bend left=40] node[right]{$m$}
  (a.west); %Z instruct Alice to send m
  \draw[->] (a) to node[above]{$m$} (f);%Alice gives m to F
  \draw[->] (k) to node[above]{$k$} (s);
  \draw[->] ($(f.south)+(-1,0)$) to node[left]{$m$}
  ($(s.north)+(-1,0)$); %S gets m from F
  \draw[->] ($(s.south)+(-1,0)$) to node[left]{$(m,t)$}
  ($(z.north)+(-1,0)$); %S forwards (m,t) to Z
  \draw[->] ($(z.north)+(1,0)$) to node[right]{$(m',t')$}
  ($(s.south)+(1,0)$); %S gets (m',t') from Z
  \draw[->] ($(s.north)+(1,0)$) to node[right]{$m'$}
  ($(f.south)+(1,0)$);%
  \draw[->] (f) to node[above]{$m,\bot$} (b);%
  \draw[->] (b.east) to[bend left=40] node[left] {$m',\bot$}
  (z.east); %Z gets Bob's output

  \end{tikzpicture}
  \quad
\begin{tikzpicture}[scale=.9,thick,every node/.style={minimum height=.6cm}]
  \scriptsize
   
  % Environment, Dummy attacker, WCA, Alice and Bob.
  \draw (0,-3) node[draw,minimum width=6cm] (z) {$\mZ$};%
  \draw (-2,0) node[fill=black!10,minimum width=3cm,minimum
  height=.8cm]{};%
  \draw (2,0) node[fill=black!10,minimum width=3.1cm,minimum
  height=.8cm]{};%
  \draw (-3,0) node[draw,minimum width=1cm,fill=white](a) {Alice};%
  \draw (3,0)node[draw,minimum width=1cm,fill=white](b) {Bob};%
  \draw (0,-1.5)node[draw,minimum width=3cm](e) {$\mA$};%
  \draw (0, 1)node[draw,rounded rectangle, minimum width=1.3cm](k)
  {Key};%
  \draw (-1,0)node[draw,rectangle,minimum width=1cm,fill=white](T)
  {TAG};%
  \draw (1,0)node[draw,rectangle,minimum width=1cm,fill=white](V)
  {VRFY};%
  \draw (0,.5) node[draw,dotted,minimum width=3cm,minimum
  height=1.7cm]{} +(1.2,.75)node{WCA}; %

  \tiny
  \draw[->] (z.west) to[bend left=40] node[right]{$m$}
  (a.west); %Z instruct Alice to send m
  \draw[->] (a) to node[above]{$m$} (T);%Alice gives m to TAG
  \draw[->] (T.south) to node[left]{$(m,t)$}
  ($(e.north)+(-1,0)$);% A gets (m,t) from TAG
  \draw[->] ($(e.south)+(-1,0)$) to node[left]{$(m,t)$}
  ($(z.north)+(-1,0)$); % A forwards (m,t) to Z
  \draw[->] ($(z.north)+(1,0)$) to node[right]{$(m',t')$}
  ($(e.south)+(1,0)$); % A gets (m',t') from Z
  \draw[->] ($(e.north)+(1,0)$) to node[right]{$(m',t')$}
  (V.south);% A passes (m',t') to VRFY
  \draw[->] (V) to node[above]{$m',\bot$} (b);%
  \draw[->] (b.east) to[bend left=40] node[left] {$m',\bot$}
  (z.east); % Z gets Bob's output

  % Alice and Bob get k from key source
  \draw[->] (k.south) to node[above right,at end]{$k$} (T.north);
  \draw[->] (k.south) to node[above left,at end]{$k$} (V.north);
      
\end{tikzpicture}
  \end{center}
  \caption{On the left is the ideal case: the ideal functionality
    $\mF$ and simulator $\mS$ complete with key input.  On the right
    is the real case: the WCA scheme and an adversary $\mA$.  The
    environment $\mZ$ wants to distinguish between the two given all
    the input and output from the system.}
  \label{fig:setup1}
\end{figure}
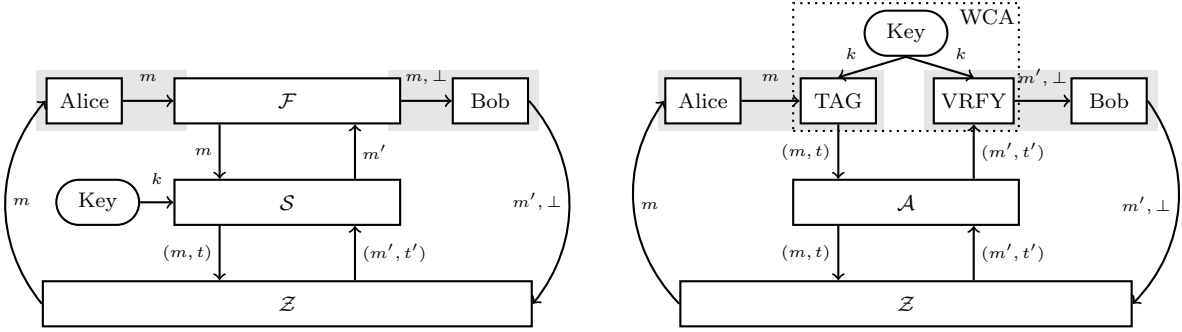

%%%%%%%%%%%%%%%%%%%%%%%%%%%%%%%%%%%%%%%%%%
%%%%%%%%%%%%%% New Lemma / Theorem %%%%%%%%%%%%%%%
%%%%%%%%%%%%%%%%%%%%%%%%%%%%%%%%%%%%%%%%%%
\begin{theorem} \label{thm-WCA-uc} \emph{(Indistinguishability)} No
  distinguisher $\mZ$ can distinguish between the two
  cases\nopagebreak

  \emph{(a)} it is interacting with $\mA$ and participants running the
  WCA scheme based on $\epsilon$-ASU$_2$ hashing using
  $\epsilon'$-perfect authentication key, or\nopagebreak

  \emph{(b)} it is interacting with $\mS$ and participants running
  $\mF$\nopagebreak 

  \noindent
  except with probability $\epsilon+\epsilon'$.
\end{theorem}

\begin{proof}
  In the proof, the message given to Alice is denoted $X$ and its
  distribution is in control of the environment $\mZ$.  The
  authentication key $K$ is used to select $f_K$ that in turn is used
  to generate the tag. The key distribution is not in control of
  $\mZ$, and has $\epsilon'$ trace distance to the uniform.  The
  corresponding output message-tag pair is denoted $Y$. The channel
  output is denoted $Y'$ and is again in control of $\mZ$. The output
  of the real and ideal functionality is denoted $\Xtilde$ and
  $\Xhat$, respectively and take values in $\mM\cup\{\bot\}$.  Thus,
  the environment $\mZ$ has access to the joint random variables
  $XYY'\Xtilde$ in the real case and $XYY'\Xhat$ in the ideal case. In
  both cases, $\mZ$ is in control of $X$ and $Y'$. The random variable
  $Y$ has an identical distribution (conditioned on the value of $X$)
  in both cases, so distinguishing the two systems can only be done
  from the output $\Xtilde$ or $\Xhat$, if the output is different
  from $X$ and also not $\bot$. This is only possible in the real
  implementation, and the probability of this is
  $\Pr\{\Xtilde\neq\bot\wedge\Xtilde\neq X\}$. This can also be
  studied through the trace distance between the two distributions
  \begin{equation}\label{eq:trace-WCA}
    \begin{split}
      \delta(P_{XYY'\Xtilde},P_{XYY'\Xhat}) =
      \frac{1}{2}\sum_{m,y,y',x'}
      \Big|P_{XYY'\Xtilde}\big(m,y,y',x'\big)-
      P_{XYY'\Xhat}\big(m,y,y',x'\big)\Big|.
    \end{split}
  \end{equation} 
  Above, the index $x'$ runs over $\mM\cup\{\bot\}$. Since the real
  and ideal cases are indistinguishable if $m=x'$, the above sum
  simplifies to the terms where $m\neq x'$. Furthermore, if $m\neq x'$
  the ideal functionality $\mF$ always outputs $\bot$.  We can
  therefore change the name of the index to $m'$ since it now runs
  only over $\mM$, and we find that the trace distance equals
  $\Pr\{\Xtilde\neq\bot\wedge\Xtilde\neq X\}$, because
  \begin{equation*}
    \label{eq:trace-WCA-2}
    \begin{split}
      \delta&(P_{XYY'\Xtilde},P_{XYY'\Xhat})= \sum_{m,y,y',m'\neq m}
      P_{XYY'\Xtilde}\big(m,y,y',m'\big)
      =\Pr\{\Xtilde\neq\bot\wedge\Xtilde\neq X\}\\
      &= \sum_{m,t,t',m'\neq m} % 2
      P_{X}(m)P_{Y|X}\big((m,t)|m\big)P_{Y'|XY}\big((m',t')\big|m,(m,t)\big)
      P_{\Xtilde|XYY'}\big(m'|m,(m,t),(m',t')\big)\\
      &= \sum_{m,t,t',m'\neq m} % 2
      P_{X}(m)\Pr\big\{h_K(m)=t\big\}P_{Y'|Y}\big((m',t')\big|(m,t)\big)
      \Pr\big\{h_K(m')=t'|h_K(m)=t\big\}\\
      &=\sum_{m,t,t',m'\neq m} P_{X}(m) %5
      P_{Y'|Y}\big((m',t')\big|(m,t)\big)\Pr\{f_K(m')=t' \,\wedge\,
      f_K(m)=t\}.
    \end{split}
  \end{equation*}
  Now, the simple bound $\Pr\{f_K(m')=t'\,\wedge\,
  f_K(m)=t\}\le\epsilon/|\mT|+\epsilon'$ (from
  Lemma~\ref{lemma:subset}) only gives
  \begin{equation*}
    \begin{split}
      \delta(P_{XYY'\Xtilde},P_{XYY'\Xhat})
      &=\sum_{m,t,t',m'\neq m} P_{X}(m)
      P_{Y'|Y}\big((m',t')|(m,t)\big)\Pr\{f_K(m')=t' \,\wedge\,
      f_K(m)=t\} \\
      &\le \sum_{m,t,t',m'\neq m}
      P_{X}(m) P_{Y'|Y}\big((m',t')|(m,t)\big)
      \Big(\frac{\epsilon}{|\mT|}+\epsilon'\Big)=\epsilon+|\mT|\epsilon',
    \end{split}
  \end{equation*}
  and that is insufficient for our purposes. This occurs for the same
  reason as the high bounds in Theorems~\ref{thm-WCA}
  and~\ref{thm-WCA-cor}: the upper bound for the individual terms
  \emph{is} this high, but the bound is not reached for all $(m,t)$.
  Here, we can do better by bounding the expression
  \begin{equation*}
    \sum_{t,t',m'\neq m}
    P_{Y'|Y}\big((m',t')|(m,t)\big)\Pr\{f_K(m')=t' \,\wedge\, f_K(m)=t\}
  \end{equation*}
  instead of the individual terms. The probability
  $P_{Y'|Y}\big((m',t')|(m,t)\big)$ corresponds to the adversary's
  attack strategy: given a message-tag pair on the input to the
  channel, choose what to substitute as output from the channel. If
  the adversary uses a deterministic attack, meaning that $(m',t')$ are
  functions of $(m,t)$, we immediately obtain
  \begin{equation*}
    \begin{split}
      \sum_{t,t',m'\neq m}
      &P_{Y'|Y}\big((m',t')|(m,t)\big)\Pr\{f_K(m')=t'\,\wedge\,f_K(m)=t\}
      \\&
      = \sum_{t}\Pr\{f_K\big(m'(m,t)\big)=t'(m,t) \,\wedge\, f_K(m)=t\}\\
      &= \Pr\Big[\bigcup_t\big\{f_K\big(m'(m,t)\big)=t'(m,t) 
      \,\wedge\, f_K(m)=t\big\}\Big]\\
      &\le|\mT|\Big(\epsilon\frac{|\mK|}{|\mT|}\Big)\frac1{|\mK|}+\epsilon_1
      =\epsilon+\epsilon_1.
    \end{split}
  \end{equation*}
  The sum can be rewritten as the probability of a union because the
  events are disjoint, and the inequality is obtained from
  Lemma~\ref{lemma:subset}. The remaining average over $m$ has no
  effect on the bound. 
  
  If the adversary has a randomized attack, we can introduce an
  auxiliary probability space $(\Omega,\mF,\mu)$ for the random
  variable $Y'=(X',T')$, where $\Omega$ is the sample space, $\mF$ is
  the $\sigma$-algebra of events, and $\mu$ is the probability
  measure. Using the indicator function $\chi$ we can write
  \begin{equation}
    P_{Y'|Y}\big((m',t')|(m,t)\big)
    =\int_\Omega\chi_{\{\omega\in\Omega:Y'(m,t,\omega)=(m',t')\}}(\omega)\,d\mu.
  \end{equation}   
  We note that for each fixed sample $\omega$, the attack is
  deterministic.  The above approach now gives
  \begin{align*}
    \sum_{t,t',m'\neq m}&
    P_{Y'|Y}\big((m',t')|(m,t)\big)\Pr\{f_K(m')=t' \,\wedge\, f_K(m)=t\}\\
    &\,=\sum_{t,t',m'\neq m}\int_\Omega\chi_{\{\omega\in\Omega:
      Y'(m,t,\omega)=(m',t')\}}(\omega)\,d\mu\Pr\{f_K\big(m'\big)
    =t' \,\wedge\, f_K(m)=t\}\\
    &\,=\int_\Omega\sum_{t}\Pr\big\{f_K\big(X'(m,t,\omega)\big)
    =T'(m,t,\omega) \,\wedge\, f_K(m)=t\big\}\,d\mu\\
    &\le \int_\Omega \,\epsilon+\epsilon'\,d\mu = \epsilon+\epsilon'.
  \end{align*}
  Again, the remaining average over $m$ has no effect on the
  bound.\hfill$\square$
\end{proof}

Now, the UC security of the WCA scheme with a partially known 
key follows immediately.
\begin{corollary} \label{cor-WCA-uc}
\emph{(UC security)}
Consider the WCA scheme based on $\epsilon$-ASU$_2$ hashing. 
Assume that the authentication key $k$ is $\epsilon'$-perfect. 
Then the WCA scheme is $\epsilon+\epsilon'$-UC-secure. 
\end{corollary}

%Thus, if the authentication key $k$ has $\epsilon'$ trace distance 
%to the uniform, then the WCA scheme is indistinguishable from the 
%authentic channel, except with probability $\epsilon + \epsilon'$, 
%and hence $\epsilon+\epsilon'$-UC-secure. 

\section{Conclusions}
\label{sec:conclusions}
We have presented a detailed security analysis of Wegman-Carter
authentication with failure probability $\epsilon$, in the case of
partially known key whose distribution is $\epsilon'$ trace distance
from the uniform distribution.  We proved tight upper bounds for the
adversary's success probability of breaking the scheme with
impersonation and substitution attacks in the information-theoretic
setting, with success probability upper bounded by $1/|\mT|+\epsilon'$
and $\epsilon+|\mT|\epsilon'$, respectively. The latter is
substantially higher than expected, but we give an example that
reaches the bound, meaning that the bound is sharp. Also in terms of
trace distance, a similar increase can be noted. The best possible
upper bound to the trace distance after having seen a valid
message-tag pair is $|\mT|\epsilon'$; the same example tells us that
this bound is sharp.

Since the bounds we obtained are substantially higher than what one
would expect, we also analyze whether the scheme is secure in terms of
witness indistinguishability.  Despite the high success probability
bound and increase in trace distance, we prove that the authentication
under study is indeed indistinguishable from the ideal functionality,
except with probability less than $\epsilon+\epsilon'$.  We provide a
direct proof for the case of partially known key, without using the
composability theorem. Naturally, UC security of the scheme with
partially known key follows from our proof of indistinguishability.

These results seem to contradict each other, but they do not. The
first should be understood as pointing out that the attacker will have
high success probability in some rounds, after having seen a valid
message-tag pair. The second shows that this happens seldom enough to
retain the expected security. The important lesson is that the
attacker can refrain from performing an active attack, if the success
probability is low after having seen a valid message-tag pair. This is
because she can calculate her success probability from available
knowledge on the key and the additional information obtained from a
valid message-tag pair. In essence she does not need to reveal herself
at each attempt to break the system, but needs only take this risk
when the success probability is high.  The security parameters should
not be read as ``the probability that an attacker is revealed, in each
attack'' but rather ``the probability that the system is broken, in
each round.'' It is important to keep this in mind when using this
type of authentication, and of course, the size of the security
parameters $\epsilon$ and $\epsilon'$ should be chosen accordingly.

\bibliographystyle{splncs03}
\bibliography{wca}

\end{document}